\documentclass[a4paper,UKenglish,cleveref, autoref, thm-restate]{lipics-v2021}

\hideLIPIcs

\usepackage{breqn}  
\usepackage{mathtools}
\usepackage{csquotes}
\usepackage{enumerate}
\usepackage{amsthm}
\usepackage{amsmath}
\usepackage{amsfonts}
\usepackage{amssymb}
\usepackage{faktor}
\usepackage{bm}
\usepackage{booktabs}
\usepackage{graphicx}
\usepackage[clock]{ifsym}

\title{On Variable-Bounded Non-Linear Expansions of Presburger Arithmetic}

\author{Piotr {Bacik}}{University of Oxford, UK \and Max Planck Institute for Software Systems, Saarland Informatics Campus, Germany}{piotr.bacik@stcatz.ox.ac.uk}{https://orcid.org/0009-0006-0248-3204}{}

\author{Joris {Nieuwveld}}{University of Oxford, UK }{jnieuwve@mpi-sws.org}{https://orcid.org/0009-0002-0339-1230}{}

\author{Jo\"el {Ouaknine}}{Max Planck Institute for Software Systems, Saarland Informatics Campus, Germany}{joel@mpi-sws.org}{https://orcid.org/0000-0003-0031-9356}{}

\author{Mihir {Vahanwala}}{Max Planck Institute for Software Systems, Saarland Informatics Campus, Germany}{mvahanwa@mpi-sws.org }{https://orcid.org/0009-0008-5709-899X}{}

\author{Madhavan {Venkatesh}}{Max Planck Institute for Software Systems, Saarland Informatics Campus, Germany}{madhavan@mpi-sws.org}{https://orcid.org/0000-0002-2698-367X}{}

\author{Emil Rugaard {Wieser}}{Max Planck Institute for Software Systems, Saarland Informatics Campus, Germany}{ewieser@mpi-sws.org}{https://orcid.org/0009-0000-4884-3953}{}

\authorrunning{P. Bacik et al.}

\Copyright{Piotr Bacik, Joris Nieuwveld, Joël Ouaknine, Mihir Vahanwala, Madhavan Venkatesh, Emil Rugaard Wieser}

\date{September 2025}

\usepackage{accents}

\newlength{\dhatheight}

\usepackage{bbm}

\nolinenumbers


\acknowledgements{P. Bacik is supported by EPSRC grant EP/X033813/1. P. Bacik, J. Ouaknine, M. Vahanwala, M. Venkatesh, and E. R. Wieser are supported by ERC grant DynAMiCs (101167561) and DFG grant 389792660 as part of \href{https://perspicuous-computing.science}{TRR 248}. J. Nieuwveld is supported by the Glasstone Benefaction, University of Oxford [Violette and Samuel Glasstone Research Fellowships in Science 2025]. J. Ouaknine is also affiliated with Keble College, Oxford as
\href{https://www.emmy.network/}{emmy.network} Fellow.}

\category{}

\relatedversion{}

\newcommand{\Q}{\mathbb{Q}}

\counterwithin{equation}{section}

\EventEditors{Claudia Faggian and Joost-Pieter Katoen}
\EventNoEds{2}
\EventLongTitle{41st Annual Symposium on Logic in Computer Science (LICS 2026)}
\EventShortTitle{LICS 2026}
\EventAcronym{LICS}
\EventYear{2026}
\EventDate{July 20--23, 2026}
\EventLocation{Lisbon, Portugal}
\EventLogo{}
\SeriesVolume{380}
\ArticleNo{73}

\begin{document}

%

\maketitle
\begin{abstract}
We consider expansions of Presburger arithmetic with families of
monadic polynomial predicates. (Examples of such predicates are the
set of perfect squares, or the set of integers of the form $2n^3-5n+3$,
etc.) Although the full attendant first-order theories are well known
to be undecidable, very
little is known when one restricts the number of variables. In the case of
single-variable theories, we obtain positive results for the following
two families of predicates: (i) for perfect fixed powers, decidability of
the corresponding theory follows from the solvability of hyperelliptic
Diophantine equations; and (ii) for polynomials of degree at most
three, we establish decidability by relying on the low genus of the
resulting algebraic curves. Finally, we discuss limitations and
hardness results (via encodings of longstanding open Diophantine
problems) as soon as any of the above restrictions are lifted.
\end{abstract}

\section{Introduction}
Presburger arithmetic was introduced and proven decidable in 1929 as a preliminary step towards Hilbert's goal of mechanising all of number theory, and in particular algorithmically determining the satisfiability of arbitrary Diophantine equations. Unfortunately, the famous works of G\"odel, Church, and Turing in the 1930s brought the Hilbert program to a screeching halt, and Matiyasevich dealt the final blow in 1970 by proving, building on a large body of work by himself and others, that solving polynomial equations over the integers was in general algorithmically infeasible; in other words, that Hilbert's tenth problem was undecidable.

Somewhat paradoxically, the demise of Hilbert's program did not dampen the scientific community's appetite for investigating the decidability of various logical theories of arithmetic and beyond: research into non-linear expansions and fragments of Presburger arithmetic, for example, remains a topic of active interest; see, for instance, the surveys~\cite{point2000decidable,bes2002survey,Haase18}, as well as the recent papers~\cite{Mansutti23,Mansutti24,Mansutti25}. A 2022 breakthrough by Hieronymi and Schulz shows that expanding Presburger arithmetic by two or more power predicates over multiplicatively independent bases leads to undecidability~\cite{HS22}; in other words, for example, 
the first-order theory $\mathit{FO} \langle   \mathbb{Z}; 0, 1, +, <, 2^{\mathbb{N}}, 3^{\mathbb{N}}\rangle$ is undecidable, where $2^{\mathbb{N}}$ and $3^{\mathbb{N}}$ stand for the sets of powers of $2$ and powers of $3$, respectively. Decidability can however be recovered when restricting to the \emph{existential} fragment, viz.\ $\exists \mathit{FO} \langle   \mathbb{Z}; 0, 1, +, <, 2^{\mathbb{N}}, 3^{\mathbb{N}}\rangle$~\cite{karimov-pres}. (The problem remains wide open when three or more power predicates are simultaneously in play.)

Let us turn to monadic \emph{polynomial} predicates,
i.e., sets of the form $\mathcal{R} = \{ f(u) \mid u \in \mathbb{Z} \}$, where $f \in \mathbb{Q}[u]$ is an integer-valued polynomial with rational coefficients.\footnote{Note that, whilst all polynomials with integer coefficients are automatically integer valued over $\mathbb{Z}$, the converse does not hold; consider, for example, the polynomial $f(u) = \binom{u}{2} = \frac{u^2+u}{2}$.}
It is folklore that, whenever $\mathcal{R}$ corresponds to a polynomial of degree at least $2$, the theory $\mathit{FO} \langle   \mathbb{Z}; 0, 1, +, <, \mathcal{R}\rangle$ is automatically undecidable, via a simple encoding of multiplication within.
In the 1970s, B\"uchi considered specifically the case of perfect squares, i.e., the predicate $\mathbf{Z}^2 := \{ u^2 \mid u \in \mathbb{Z}\}$, and asked about the decidability of the \emph{existential} fragment $\exists \mathit{FO} \langle \mathbb{Z}; 0, 1, +, <, \mathbf{Z}^2\rangle$. As we describe in greater detail in Sec.~\ref{sec:hardness}, B\"uchi in fact formulated a conjecture implying undecidability of this theory; a proof of B\"uchi's conjecture was only recently announced by 
Xiao~\cite{xiao2025buchi}, finally establishing undecidability of the corresponding logical theory after some five decades! Note that Xiao's proof only concerns the perfect-square predicate, and the general question of the decidability of $\exists \mathit{FO} \langle \mathbb{Z}; 0, 1, +, <, \mathcal{R}\rangle$, where $\mathcal{R}$ is an arbitrary non-linear polynomial predicate, remains open.

In addition to restricting the number and use of quantifiers, another classical means of attempting to recover decidability involves bounding the number of variables; standard references on bounded-variable logics include~\cite{Ott96,Gro98,GKV97,GO99,pra23}. 

We are now in a position to describe our main contributions. We focus on variable- and quantifier-bounded expansions of Presburger arithmetic with families of monadic polynomial predicates. Since variables and quantifiers are now in short supply, we expand our base signature to maximise expressiveness and flexibility,\footnote{For example, the subtraction operator is typically not included in the signature of Presburger-arithmetic theories, since a term such as $-x$ can be recovered through existential quantification: $\exists y \, . \, y + x =0$. Likewise, modular-arithmetic constraints are usually implicit: $x$ is an even number if and only if $\exists y \, . \, x = y+y$, etc.} by considering the following, where $k$ is the bound on the number of allowable distinct variables and the binary relation symbol
$\equiv_m$ refers to congruence modulo $m$:
\begin{itemize}
\item $\mathit{FO}^k\langle   \mathbb{Z}; 0, 1, +, -, <, (\equiv_m)_{m \geq 2} \rangle$ 
denotes the first-order fragment with no restrictions on quantifiers;
\item $\exists\mathit{FO}^k\langle   \mathbb{Z}; 0, 1, +, -, <, (\equiv_m)_{m \geq 2} \rangle$ denotes the existential fragment;
\item $\mathit{SMT}^k\langle   \mathbb{Z}; 0, 1, +, -, <, (\equiv_m)_{m \geq 2} \rangle$ denotes the \emph{satisfiability modulo theories} fragment, i.e., 
existential formulas in prenex normal form:
$\exists x_1, \ldots, x_k \, . \, \varphi(x_1,\ldots,x_k)$, with
$\varphi(x_1,\ldots,x_k)$ quantifier free.
\end{itemize}

Somewhat surprisingly, even restricting to a \emph{single} variable (i.e., $k=1$) immediately leads to well-known open problems: let $\mathcal{R}_1$ and $\mathcal{R}_2$ be predicates corresponding to polynomials $f_1$ and $f_2$; in general the decidability of whether there are integers $u$ and $v$ such that $f(u) = g(v)$ --- a severely restricted instance of Hilbert's tenth problem --- is open. But such a query is easily encodable within the bare theory 
$\mathit{SMT}^1\langle   \mathbb{Z}; \mathcal{R}_1, \mathcal{R}_2 \rangle$, by asking for the truth value of $\exists x \, . \, \mathcal{R}_1(x) \wedge \mathcal{R}_2(x)$.
Even in the case of a \emph{single} predicate $\mathcal{R}$ (with underlying polynomial $f$), decidability remains open. Consider, for arbitrary integer constants $a,b,c,d$, 
the sentence $\exists x \, . \, \mathcal{R}(ax+b) \wedge \mathcal{R}(cx+d)$, which is readily expressible in
$\mathit{SMT}^1\langle   \mathbb{Z}; 0, 1, +, - ,\mathcal{R} \rangle$. This formula asserts the existence of integers $u$ and $v$ such that $cf(u)+ad = af(v)+cb$; however in general it is not known whether such Diophantine equations can always be solved.

Our main positive results exclusively concern single-variable theories (in which case the first-order, existential, and SMT fragments essentially all coincide). We establish decidability of the following:
\begin{enumerate}
\item Single-variable expansions of Presburger arithmetic by arbitrarily many polynomial predicates corresponding to perfect fixed powers (Thm.~\ref{thm:1-pres-Zk}):
$$\mathit{FO}^1 \langle   \mathbb{Z}; 0, 1, +, -, <, (\equiv_m)_{m \geq 2}, 
    (\mathbf{Z}^k)_{k \geq 2}\rangle \, .$$
\item Single-variable expansions of Presburger arithmetic by arbitrarily many polynomial predicates of degree at most $3$ (Thm.~\ref{thm:2-pres-sqcu}):
$$\mathit{FO}^1 \langle   \mathbb{Z}; 0, 1, +, -,
    <, (\equiv_m)_{m \geq 2}, (\mathcal{R}_i)_{i}\rangle \, . $$
\end{enumerate}    
    
Amongst other ingredients, these theorems are obtained by making use of deep results on the solvability of hyperelliptic Diophantine equations and equations corresponding to algebraic curves of low genus. Or more precisely, assertions involving the \emph{positive} use of polynomial predicates translate to number-theoretic or algebro-geometric constraints, which can then be analysed using the relevant mathematical machinery,
whereas assertions involving \emph{negated} predicates are handled through measure-theoretic (or ``density''-type) arguments.

It is perhaps useful at this stage to provide a few concrete examples of the kinds of statements that can be expressed in instances of such single-variable theories:
\begin{itemize}
\item \emph{There is no triangular number\footnote{A triangular number is a positive integer of the form $n(n+1)/2$, for some integer $n$.} larger than $1$ that is a perfect cube.} This assertion is due to Fermat (see~\cite[Prop.~0.9]{kato00}, where the authors also note, ``it is very difficult to prove [this] proposition by hand without using any significant tools. In attempting to prove [it] we are naturally led to profound
mathematics.''). Writing $\mathcal{T}(x)$ to express the fact that $x$ is a triangular number, Fermat's statement is equivalent to the falsity of the sentence $\exists x \, . \, x > 1 \wedge \mathcal{T}(x) \wedge \mathbf{Z}^3(x)$. This formula is readily seen to belong to theories from the second of our two decidable classes above, since the polynomial predicates at play, namely $\mathcal{T}$ and $\mathbf{Z}^3$, have degrees $2$ and $3$ respectively.

\item \emph{The largest cube in the sequence of Fibonacci numbers is $8$.} 
Here we make use of a well-known result of Gessel to the effect that $n$ is a Fibonacci number if and only if either $5n^2+4$ or $5n^2-4$ is a perfect square~\cite[Prob.~H-187]{gessel72}. The desired statement is therefore equivalent to the non-existence of an $n>8$ such that $5n^2+4$ or $5n^2-4$ is a perfect square, and moreover such that $n$ is a perfect cube; letting $x$ denote $n^2$, we can express this as the negation of
$\exists x \, . \, \mathbf{Z}^2(x) \wedge x>64 \wedge (\mathbf{Z}^2(5x+4) \vee 
\mathbf{Z}^2(5x-4)) \wedge \mathbf{Z}^6(x)$. As written, this formula belongs to theories from the first of our two decidable classes; observe, however, that for any integer $n$,
whenever $n^2$ is perfect perfect cube, then so is $n$, and thus the conjunct 
$\mathbf{Z}^6(x)$ in our sentence can safely be replaced by $\mathbf{Z}^3(x)$, yielding a new formula that belongs to theories from both of our classes.

\item For any fixed power $d>3$, one can similarly assert that 
\emph{there is no Fibonacci number greater than $1$ that is a perfect $d$-th power.} This is a special case of a famous 2006 result of Bugeaud, Mignotte, and Siksek~\cite{bugeaud06}, according to which the
only perfect powers in the Fibonacci sequence are $0$, $1$, $8$, and $144$. Here the corresponding formulas belong to theories from the first decidable class.

\item Catalan's conjecture, stating that $9$ and $8$ are the only perfect powers with difference exactly $1$, was open for $150$ years, and only proven in 2002~\cite{catalan}. For any fixed $u$-th and $v$-th power, we can spell out the special case that 
\emph{there are no $n, m > 2$ such that $n^u - m^v = 1$}, as the negation of the sentence 
$\exists x \, . \, x>8 \wedge \mathbf{Z}^v(x) \wedge \mathbf{Z}^u(x+1)$.

\item Pillai's conjecture, along with the ``Generalised Tijdeman problem'', extend Catalan's conjecture, and remain open (see~\cite{pillai-conj}); the assertion is that, for every $k>1$, there exist only finitely many pairs of perfect powers with difference exactly $k$: $n^u - m^v = k$. 
We can readily express (and therefore substantiate) within our decidable formalism any special instance of this conjecture in which $u$, $v$, and $k$ are fixed. More precisely, for fixed positive integers $u,v,k,T$, we can write 
a formula $\exists x \, . \,\varphi(x,u,v,k,T)$, making use of the single variable $x$, expressing the fact that $n^u - m^v = k$ has some integer solution with $n, m > T$. If Pillai's conjecture is true, then by choosing $T$ sufficiently large, $\varphi$ eventually becomes false, substantiating the special case in question.

\end{itemize}

Finally, in Sec.~\ref{sec:hardness}, we complement our decidability results by 
establishing undecidability of expansions involving the perfect-square predicate $\mathbf{Z}^2$ when several variables are allowed (Thm.~\ref{thm:undec}).
Let us meanwhile conclude this introduction by briefly pointing out some of the formidable obstacles that stand in the way of extending our decidability results to two- or three-variable fragments involving merely the perfect square or perfect cube predicates. 

We start by considering \emph{perfect Euler bricks}, i.e., rectangular boxes with integer sides, all of whose diagonals are moreover also integers. The existence of perfect Euler bricks has famously remained open for over three centuries; to date none has ever been found, and no-one has been able to show that they cannot exist. Note, however, that the existence of a perfect Euler brick is straighforwardly encodable within $\mathit{SMT}^3\langle   \mathbb{Z}; 0, +, \mathbf{Z}^2\rangle$, thanks to the Pythagorean theorem; decidability of this three-variable fragment should therefore be considered squarely out of reach. On the other hand, decidability of two-variable fragments of expansions of Presburger arithmetic by the perfect-square predicate remains open and a fascinating avenue for further research.

The situation with $\mathit{SMT}^2\langle \mathbb{Z}; 0, 1, +, -,  \mathbf{Z}^3\rangle$, i.e., the case of purely existential two-variable sentences involving the perfect-cube predicate, also appears rather hopeless at present. That is because this fragment can encode arbitrary instances of the famous \emph{sum-of-three-cubes problem},\footnote{See 
\href{https://en.wikipedia.org/wiki/Sums_of_three_cubes}%
{\texttt{https://en.wikipedia.org/wiki/Sums_of_three_cubes}}.}
in which one asks, for a given positive integer $n$, whether there are integers $u,v,w$ such that $u^3+v^3+w^3=n$. Although the answer is known for several values of $n$, infinitely many instances remain open, with $n=114$ being the smallest at the time of writing.

\section{Technical Preliminaries}
\subsection*{Integer-Valued Polynomials}
We refer to \cite{IntValPoly} for basic properties of univariate integer-valued polynomials. In particular, a univariate integer-valued polynomial of degree $k$ can uniquely be written as an integer linear combination $\sum_{r=0}^{k} f_r \cdot \binom{x}{r}$, where $\binom{x}{r}$ is the binomial polynomial $\frac{x(x-1)\cdots(x-r+1)}{r!}$ and is always integer valued. We have that $\binom{x}{0} =1$ and $\binom{x}{1} = x$.

\subsection*{Linear Recurrence Sequences}
Linear recurrence sequences enable us to describe solution sets of certain Diophantine equations that arise in our analysis.

A linear recurrence relation over $\mathbb Q$ is an equation of the form
\begin{equation}\label{eq:LRS_def}
    u_{n+d} = a_{1} u_{n+d-1} + \dots + a_d u_n   \, ,
\end{equation}
where $a_1, \dots, a_{d} \in \mathbb Q$ and $a_d \neq 0$. The initial values $u_0,u_1, \dots, u_{d-1} \in \mathbb{Q}$ and \eqref{eq:LRS_def} together uniquely define a sequence of rationals $\langle u_n \rangle_{n=0}^\infty$ as well as a bi-sequence $\langle u_n \rangle_{n=-\infty}^\infty$. We refer to the former as a \emph{linear recurrence sequence} (LRS) and the latter as a \emph{linear recurrence bi-sequence} (LRBS)\@. The smallest integer $d$ for which a sequence obeys a relation of the form \eqref{eq:LRS_def} is the \emph{order} of the sequence.

Note that if $a_1, \dots, a_d, u_0, \dots, u_{d-1} \in \mathbb Z$ and $a_d = \pm 1$, then \eqref{eq:LRS_def} implies that the LRBS $\bm u = \langle u_n \rangle_{n=-\infty}^\infty$ is entirely contained in~$\mathbb Z$. In fact, an old result of Fatou \cite{Fatou_1904} (see also \cite[Chap. 7]{Berstel_Reutenauer_2010}) implies that an LRBS $\bm u$ satisfying \eqref{eq:LRS_def}, with $a_1,\dots,a_d,u_0,\dots,u_{d-1} \in \mathbb Z$, is contained in $\mathbb Z$ if and only if $a_d = \pm 1$. In this case, we say that $\bm u$ is \emph{reversible}.

Note that $\bm u$ has an exponential-polynomial form
\begin{equation*}
    u_n = \sum_{i=1}^s P_i(n) \lambda_i^n
\end{equation*}
where $\lambda_i$ are the \emph{characteristic roots}, that is, roots of the \emph{characteristic polynomial} 

\begin{equation*}
    g(x) = x^d - a_1 x^{d-1} - \dots - a_d \, ,
\end{equation*}
and $P_i$ are polynomials with algebraic coefficients with degree one less than the multiplicity of $\lambda_i$ as a root of $g$. We say that $\bm u$ is \emph{simple} if none of the roots of $g$ are repeated, which in turn is equivalent to each $P_i$ being constant. 

\subsection*{Diophantine Equations}
A Diophantine equation is a multivariate polynomial equality with integer coefficients for which one seeks integer solutions. In this section we detail some Diophantine equations that arise later, and how to obtain their solution sets.
\begin{definition}
    A \emph{Pell equation} is a Diophantine equation of the form $w^2 - n z^2 = 1$, where the coefficient $n>0$ is required not to be a perfect square. The solution $(w_0, z_0)$ with $w_0,z_0>0$ which minimises $w$ is called its \emph{fundamental solution}. Diophantine equations of the form $w^2 - n z^2 = N$ (where $n>0$ is not a perfect square and $N \ne 0$) are called \emph{generalised Pell equations}.
\end{definition}

The history of Pell equations goes back to ancient times, and it is well known that the fundamental solution exists and can be computed. We refer the reader to \cite{jacobson2009solving} for a comprehensive account and modern developments. We are specifically interested in \cite[Chap.~16.3]{jacobson2009solving}, which shows that the fundamental solution $(w_0, z_0)$ is the one for 
which $w_0, z_0$ are positive and $w_0 + z_0\sqrt{n}$ is minimal.

\begin{lemma}
\label{lem:solve-general-Pell}\cite[Thm.~16.3]{jacobson2009solving}
    Consider the generalised Pell equation $w^2 - n z^2 = N$, and let $(w_0, z_0)$ be the fundamental solution of $w^2 - nz^2 = 1$. We can compute a finite set $S$ of \emph{generating} pairs $(w_i, z_i)$ such that every solution $(w', z')$ to the above generalised Pell equation satisfies $w' + z'\sqrt{n} = (w_i + z_i\sqrt{n})(w_0 + z_0\sqrt{n})^m$ for some $(w_i, z_i) \in S$ and $m \in \mathbb{Z}$.
\end{lemma}
The algebraic number $w_0 + z_0\sqrt{n}$ derived from the fundamental solution is said to be the \emph{fundamental unit} of the Pell equation.

\begin{corollary} \label{cor:pell-sol-set}
The set of solutions $(w,z)$ to a generalised Pell equation $w^2-nz^2 = N$ is obtained as a finite union of pairs of simple reversible LRBS. Moreover, these pairs of LRBS take the form
\begin{equation*}
(w_m,z_m) = \left( A_1 \varepsilon^m + A_2 \varepsilon^{-m}, B_1 \varepsilon^m + B_2 \varepsilon^{-m} \right) \, ,
\end{equation*}
where $\varepsilon$ is the fundamental unit of the generalised Pell equation, and $A_1,A_2,B_1,B_2$ are non-zero algebraic numbers.
\end{corollary}

\begin{proof}
By Lem.~\ref{lem:solve-general-Pell} every solution $(w_{i,m},z_{i,m})$ satisfies
\begin{equation}
    w_{i,m} + z_{i,m}\sqrt n = (w_i+z_i \sqrt n)(w_0+z_0\sqrt n)^m \label{eq:gen_Pell_LRBS}
\end{equation}
for some computable $m \in \mathbb Z$, fundamental solution $(w_0,z_0)$, and a finite set of pairs $(w_i,z_i)$. From \eqref{eq:gen_Pell_LRBS}, we have
\begin{align*}
    w_{i,m+t} + z_{i,m+t} \sqrt{n} = (w_0+z_0 \sqrt{n})^t (w_{i,m} + z_{i,m} \sqrt{n})
\end{align*}
for $t = 1,2$, so we may compute
\begin{align*}
    w_{i,m+2} + z_{i,m+2} \sqrt{n} -2w_0 (w_{i,m+1} + z_{i,m+1} \sqrt{n}) = -(w_{i,m} + z_{i,m} \sqrt{n}) \, ,
\end{align*}
by using that $w_0^2 - nz_0^2=1$. Thus, by the equation coefficients of $1$ and $\sqrt n$ in the above, we get
\begin{align}
    w_{i,m+2} &= 2w_0 w_{i,m+1} - w_{i,m} \label{eqn:w-rec}\quad \text{and}\\
    z_{i,m+2} &= 2w_0 z_{i,m+1} - z_{i,m} \label{eqn:z-rec}\, ,
\end{align}
so each $\langle w_{i,m} \rangle_{m=- \infty}^\infty$ and $\langle z_{i,m} \rangle_{m= - \infty}^\infty$ define reversible LRBS\@. The characteristic polynomials of \eqref{eqn:w-rec} and \eqref{eqn:z-rec} are both equal to $X^2-2w_0X+1$, which has distinct characteristic roots $(\varepsilon,\varepsilon^{-1})$, so the solutions take the exponential-polynomial form 
\begin{equation*}
    (w_{i,m},z_{i,m}) = \left( A_{1,i} \varepsilon^m + A_{2,i} \varepsilon^{-m},   B_{1,i} \varepsilon^m + B_{2,i} \varepsilon^{-m}  \right) \, ,
\end{equation*}
where $A_{1,i},A_{2,i},B_{1,i},B_{2,i}$ are algebraic numbers, as claimed. One can easily solve for these coefficients using \eqref{eq:gen_Pell_LRBS} to show that 
\begin{align*}
    A_{1,i} = \frac{w_i+z_i \sqrt n}{2} , \quad A_{2,i} = \frac{w_i-z_i\sqrt{n}}{2} , \quad B_{1,i} = \frac{w_i + z_i \sqrt{n}}{2 \sqrt{n}}  , \quad B_{2,i} = \frac{-w_i + z_i \sqrt{n}}{2 \sqrt{n}} \, .
\end{align*}
In particular, each of the coefficients is non-zero.\qedhere

\end{proof}

\begin{lemma}
    \label{lem:simul-pell-finite}
    Consider the system of simultaneous generalised Pell equations $w^2 - n_1 z_1^2 = N_1, w^2 - n_2 z_2^2 = N_2$, where $n_1 n_2$ is not a perfect square and $N_1 \ne N_2$. This system has only finitely many solutions which can moreover be effectively enumerated.
\end{lemma}
\begin{proof}
    Writing $z =n_1 n_2 z_1 z_2$, it suffices to prove that $z^2 = n_1 n_2(w^2 - N_1)(w^2 - N_2)$ has only finitely many solutions which can moreover be effectively enumerated. This is done by a direct application of \cite{Baker1969hyperelliptic} or \cite[Thm.~4.2]{Baker_1975}\footnote{The original 1975 print makes a mistake of omission in the statement of the theorem, which was subsequently corrected in later editions.}
    (see in particular the comment at the beginning of the proof, which clarifies that the case $Y^2 = c(X - \alpha_1)\cdots (X - \alpha_n)$ is also handled by the proof).
\end{proof}

We remark that algorithms to find solutions have been further refined, see e.g., \cite{Szalay2007, tzanakis}. The work of Baker \cite{Baker1969hyperelliptic} implies the following lemma for so-called hyperelliptic equations.

\begin{lemma}
    \label{lem:higher-k-finite-sol}
    Let $k \ge 2$, $j \ge 3$, $N \ne 0$, and $n$ be integers. The Diophantine equation $w^k = n z^j + N$ has only finitely many solutions, which can moreover be effectively enumerated.
\end{lemma}

\section{Fixed-Power Predicates}\label{sec:power predicates}
In this section, we prove the following theorem regarding single-variable Presburger arithmetic, where $\equiv_m$ is a binary relation symbol denoting congruence modulo $m$,
and the predicate $\mathbf{Z}^k$ is the set $\{n^k \mid n \in \mathbb{Z}\}$ of perfect $k$-th powers.
\begin{theorem}
\label{thm:1-pres-Zk}
    The theory $\mathit{FO}^1 \langle   \mathbb{Z}; 0, 1, +, -, <, (\equiv_m)_{m \geq 2}, 
    (\mathbf{Z}^k)_{k \geq 2}\rangle$ is decidable.
\end{theorem}

The core subroutine in the decision procedure solves systems of Diophantine equations of the form $w^k - nz^j  = N$. We begin by describing the pre-processing that leads to its invocation, which is summarised in the following proposition.

\begin{proposition}
    \label{prop:pre-1-pres-Zk}
    The decision problem in Thm.~\ref{thm:1-pres-Zk} Turing-reduces to deciding whether there exists $x \in \mathbb{Z}$ that satisfies a given set of constraints, which includes exactly one constraint of the form $x > c$, and constraints of the form $\mathbf{Z}^k(ax + b)$ and $\neg \mathbf{Z}^k(ax + b)$, where $a>0$.
\end{proposition}
\begin{proof}
We first prove that deciding the theory indeed reduces to a constraint satisfaction problem. Any sentence in the theory may be written in prenex normal form as $Qx \, . \,\psi$ where $Q$ is a quantifier and $\psi$ is a quantifier-free formula. Note that since $\forall x\, . \,\psi$ is equivalent to $\neg \exists x\, . \,\neg \psi$, we may reduce to deciding the truth of existential sentences, i.e., when $Q$ is $\exists$. By putting $\psi$ in disjunctive normal form, we may rewrite $\exists x \, . \,\psi$ as $\exists x\, . \bigvee_i \varphi_i$ where each $\varphi_i$ is a conjunction of literals, and this may further be rewritten as $\bigvee_i \exists x\, . \,\varphi_i$. In this way we reduce to deciding the satisfiability of a given conjunction of literals, and we proceed by analysing the constraints that may arise from literals in the theory.

By suitably rearranging and simplifying, we can assume that literals are of the form $x = c$, $x < c$, $x > c$, $x \equiv_m c$, $\mathbf{Z}^k(ax+b)$, and $\neg \mathbf{Z}^k(ax + b)$. Observe that this requires rewriting $\neg (x=c)$ as $(x<c) \vee (x > c)$, $\neg (ax \equiv_m b)$ as $\bigvee^{m-1}_{r = 0, r \ne b} ax \equiv_m r$, and $ax \equiv_m b$ as $\bigvee_{r=0, ar \equiv_m b}^{m-1} x \equiv_m r$. 

To eliminate the modular-arithmetic constraints, we use an extended version of the Chinese Remainder Theorem (see e.g., \cite[Thm.~3.12]{elem-NT-text}) to coalesce the modular-arithmetic constraints $x \equiv_{m_i} c_i$ into a single conjunct $x \equiv_M r$, or prove that they are infeasible. We now make this constraint implicit by replacing all occurrences of $x$ by $My + r$, and simplifying the resulting expressions.

We claim that we may further reduce to considering conjunctions where $x = c$ does not occur, and where the only inequality that occurs is $x > c$ for some $c \in \mathbb{Z}$. Indeed,
if a term $x = c$ does occur, we simply perform the obvious substitution, reducing to a quantifier-free formula. If the remaining conjuncts imply that $x$ is in a bounded interval, i.e., there are terms $x > c_1$ and $x < c_2$, this case is readily solved by finite inspection.  
Thus we may assume at most one inequality appears. 
If no inequalities appear, we simply case split by considering in turn 
$x < 0$, and $x = 0$, and $x >0$. If we have $x < c$, the obvious linear substitution $x \mapsto -x$ turns this term into one of the form $x > c$. This proves our claim.

We finally address the positivity of the coefficients of $x$ in the power predicates. If $k$ is odd, we can assume that in all instances (positive and negative) of $\mathbf{Z}^k(ax + b)$, the coefficient $a$ is positive, by possibly replacing $\mathbf{Z}^k(ax+b)$ by $\mathbf{Z}^k(-ax-b)$. If $k$ is even and $a < 0$ in a positive occurrence of such an atom, we have an upper bound on $x$ (as $x$ is assumed to be lower-bounded), and the conjunction can be handled trivially. If $a < 0$ for a term $\neg \mathbf{Z}^k(ax + b)$, this term always holds when $x$ exceeds a computable bound and can thus be disposed of straightforwardly. We can therefore reduce to the case where all coefficients of $x$ in the power predicates are positive. 
\end{proof}
We make a further observation: certain power constraints can be ``redundant'' in view of other power constraints. For example, consider the three constraints $\mathbf{Z}^2(x)$, $\mathbf{Z}^2(3x)$, and $\mathbf{Z}^4(16x)$. If $\mathbf{Z}^4(16x)$ holds, then $\mathbf{Z}^2(x)$ also holds and $\mathbf{Z}^2(3x)$ does not. We therefore say that $\mathbf{Z}^2(x)$ and $\mathbf{Z}^2(3x)$
are both \emph{redundant with respect to} $\mathbf{Z}^4(16x)$, since the (positive) truth of $\mathbf{Z}^4(16x)$ uniquely determines the truth values of the other two constraints.
We formalise this idea in the following definition:
\begin{definition}
The constraint $\mathbf{Z}^k(cx+d)$
is \emph{redundant with respect to} $\mathbf{Z}^j(ax+b)$ if $k \mid j$ and $ad=bc$. 
\end{definition}
In general, if $\mathbf{Z}^k(cx+d)$ is redundant with respect to $\mathbf{Z}^j(ax+b)$ then the (positive) truth of $\mathbf{Z}^j(ax+b)$ determines the truth value of $\mathbf{Z}^k(cx+d)$. Indeed, if $ax + b$ is a perfect $j$-th power, then in particular it is a perfect $k$-th power, and so is $c^k (ax + b) = ac^{k-1}(cx + d)$ (using $ad = bc$). We thus have that $cx + d$ is a perfect $k$-th power if and only if $ac^{k-1}$ is, and the latter can be effectively checked. The above discussion shows that we can identify and discard (positive or negative) constraints that are redundant with respect to some given positive constraint.

In the same vein, we define the notion of \emph{similar} constraints.
\begin{definition}
The constraint $\mathbf{Z}^j(cx + d)$
is \emph{similar} to a constraint $\mathbf{Z}^k(ax + b)$ if $ad = bc$. 
\end{definition}
Note that the notion of being similar is a transitive property.
Though similar constraints cannot be as immediately discarded as redundant constraints, we will show that a conjunction of similar positive constraints can be coalesced into a single positive constraint.

For the proof, we require the notion of \emph{$p$-adic valuation}. Recall that for a prime $p$, the $p$-adic valuation of a non-zero integer $n$, denoted $\mathfrak{v}_p(n)$, is equal to the highest power of $p$ that divides $n$, e.g., $\mathfrak{v}_2(20) = 2$. We take $\mathfrak{v}_p(0) = \infty$. The valuation $\mathfrak{v}_p$ extends to rational numbers as $\mathfrak{v}_p(m/n) = \mathfrak{v}_p(m) - \mathfrak{v}_p(n)$.

\begin{lemma} \label{lem:similar_coalesce}
Given similar constraints $\mathbf{Z}^{k_1}(a_1x+b_1), \dots, \mathbf{Z}^{k_l}(a_lx+b_l)$, either they are not simultaneously satisfiable, or we may find a constraint $\mathbf{Z}^K(Ax+B)$ that is satisfied if and only if $\mathbf{Z}^{k_i}(a_ix+b_i)$ is satisfied for all $1 \leq i \leq l$. Furthermore, $K = \mathrm{lcm}(k_1, \dots, k_l)$, where $\mathrm{lcm}$ denotes the least common multiple.
\end{lemma}
\begin{proof}
Assume that $\mathbf{Z}^{k_1}(a_1 x + b_1), \ldots, \mathbf{Z}^{k_l}(a_l x + b_l)$ are similar, and let $b/a$ with $a > 0$ be the reduced form of the rational constant $b_i/a_i$ (which is independent of $i$ as $a_ib_j = b_ia_j$ for all $1 \le i,j\le l$).

If all constraints are satisfied, we have that $\mathfrak{v}_p(a_i x + b_i) \equiv 0 \mod k_i$ for all primes $p$ and indices $i = 1, \ldots, l$. This can be rearranged as 
\begin{equation} \label{eqn:vp_mod_k_i}
\mathfrak{v}_p(ax + b) \equiv  \mathfrak{v}_p(a) - \mathfrak{v}_p(a_i) \mod k_i
\end{equation}
for each $p,i$.
Note that $\mathfrak{v}_p(a) - \mathfrak{v}_p(a_i) \equiv 0 \mod k_i$ holds for all $i$ for all but finitely many primes that divide some $a_i$ --- let us call such primes {``interesting''}. For each interesting prime $p$, we apply the (extended) Chinese Remainder Theorem \cite[Thm.~3.12]{elem-NT-text} to determine whether the constraints given by \eqref{eqn:vp_mod_k_i} for $i=1, \dots, l$ are simultaneously satisfiable, and if so, compute a residue $r_p$ such that they hold if and only if $\mathfrak{v}_p(ax + b) \equiv -r_p \mod K$, where $K = \mathrm{lcm}(k_1,\dots,k_l)$.
This is equivalent to $(ax + b)\prod_p p^{r_p}$ being a perfect $K$-th power. 
By construction, $\mathbf{Z}^K((ax + b)\prod_p p^{r_p})$ holds if and only if $\mathbf{Z}^{k_i}(a_ix +b_i)$ holds for all $1 \leq i \leq l$. 
\end{proof}

\begin{remark}
If $\mathbf{Z}^{k_1}(a_1x + b_1), \dots, \mathbf{Z}^{k_l}(a_lx+b_l)$ are similar and simultaneously satisfiable, by Lem.~\ref{lem:similar_coalesce} we may discard each $\mathbf{Z}^{k_i}(a_ix+b_i)$ to be replaced by the single constraint $\mathbf{Z}^K(Ax+B)$. We say that we \emph{coalesce} $\mathbf{Z}^{k_1}(a_1x+b_1), \dots, \mathbf{Z}^{k_l}(a_lx + b_l)$ into $\mathbf{Z}^K(Ax+B)$.
\end{remark}

In the sequel, we shall assume there are no pairs of redundant constraints or similar positive constraints: we first discard redundant constraints, then coalesce similar positive constraints into a single positive constraint, and then again discard redundant constraints. As an example, if we had $\mathbf{Z}^2(5x), \mathbf{Z}^3(4x), \neg \mathbf{Z}^6(24x)$, the last constraint becomes redundant only after the first two are coalesced into $\mathbf{Z}^6(500x)$.

Once the pre-processing step is completed, our strategy for solving the satisfiability problem consists in handling the various constraints ($x > c$, positive constraints of the form $\mathbf{Z}^k(ax + b)$, and negative constraints of the form $\neg \mathbf{Z}^k(ax + b)$) sequentially in the order given. 
More precisely, we maintain a set $S \subseteq \mathbb{Z}$ of candidate solutions, and iterate over the positive constraints. If, at a given iteration, the set $S$ of candidates is finite, then we can decide satisfiability by simply enumerating $S$. Otherwise, taking a further positive constraint into account will restrict the set of candidates to a subset that has relative density $0$ in $S$.
Let $S_{\mathsf{pos}}$ be the set of candidates after having accounted for all the positive constraints. If this set is infinite, then the overall conjunction is necessarily satisfiable. Indeed, by the same reasoning above, any negative constraint will rule out only a subset of $S_{\mathsf{pos}}$ with relative density $0$.

To formally implement this argument, we first analyse the solution sets arising from positive constraints.

\begin{proposition} \label{prop:pos_pow_sol_sets}
The solution set of all $x \in \mathbb{Z}$ satisfying $l$ non-similar constraints of the form $\mathbf{Z}^{k_i}(a_ix+b_i)$ with $a_i>0$ for $1 \leq i \leq l$ is effectively computable and has the following structure.
\begin{enumerate}
    \item If $l = 0$ then $S = \mathbb{Z}$.
    \item If $l = 1$ then either $S = \varnothing$ or $S$ is a finite union of sets of the form $f(\mathbb{Z})$ where $f$ is a polynomial of degree $k_1$ with a positive leading coefficient.
    \item If $l = 2$ then either $S = \varnothing$ or one of the following holds.
    \begin{enumerate}
        \item If $\max \{k_1,k_2\} > 2$ then $S$ is finite.
        \item If $k_1,k_2 = 2$ then $S$ is a union of finitely many simple reversible LRBS. 
    \end{enumerate}
    \item If $l \geq 3$ then $S$ is finite.
\end{enumerate}
\end{proposition}
\begin{proof}
\textbf{Case (1)} is obvious.

\textbf{Case (2)}: If $l = 1$ then $S = \varnothing$ if $b_1$ is not a $k_1$-th power modulo $a_1$, and this can be checked algorithmically by enumerating all $k_1$-th powers modulo $a_1$. Conversely, if $b_1$ is a $k_1$-th power modulo $a_1$ then there are infinitely many solutions and they can be parametrised as follows. Pick $u \in \mathbb{N}$ such that $u^{k_1} \equiv b_1 \bmod a_1$. Then consider the polynomial $f_{u,k_1,a_1} \in \mathbb{Z}[t]$ that satisfies
\begin{equation*}
    (u+ta_1)^{k_1} =  a_1f_{u,k_1,a_1}(t) + b_1 \, .
\end{equation*}
Then $f_{u,k_1,a_1}$ has degree $k_1$ and a positive leading coefficient (as $a_1 > 0$). Every $t \in \mathbb{Z}$ gives rise to a solution $x = f_{u,k_1,a_1}(t)$ of $\mathbf{Z}^k(a_1x+b_1)$. Conversely, whenever $x \in \mathbb{Z}$ is a solution, then $(a_1x+b_1)^{1/k_1}$ must be of the form $u+ta_1$ for some $u$ such that $u^{k_1} \equiv b_1 \bmod a_1$ and some $t \in \mathbb{Z}$. Therefore the solution set is exactly 
\begin{equation*}
    S = \bigcup_{\substack{0\leq u < a_1 \\ u^{k_1} \equiv b_1 \bmod a_1}} f_{u,k_1,a_1}(\mathbb Z) \, .
\end{equation*}
\textbf{Case (3)}: Suppose $l = 2$ and (without loss of generality) $k_2 \geq k_1$. We have the system
\begin{equation}\label{eq:sim_sys_1}
    a_1x + b_1 = y_1^{k_1} \:\land \:   a_2x + b_2 = y_2^{k_2} \,.
\end{equation}
By taking a linear combination, we eliminate $x$ to obtain the system
\begin{align}
    a_2 y_1^{k_1} - a_1 y_2^{k_2} &= a_2b_1 - a_1b_2 \label{eq:sim_sys_3} \\
    y_1^{k_1} \equiv_{a_1} b_1 &\wedge y_2^{k_2} \equiv_{a_2} b_2    \, .
\end{align}
By multiplying \eqref{eq:sim_sys_3} by $a_2^{k_1-1}$, and setting $w \coloneq a_2y_1$ and $z \coloneq y_2$ we obtain the system
\begin{align}
    w^{k_1} - (a_1 a_2^{k_1-1})z^{k_2} &= a_2^{k_1-1}(a_2b_1-a_1b_2) \label{eq:sim_sys_4} \\
    \bigvee_{\substack{0 \leq r_i < a_i \\ r_1^{k_1} \equiv_{a_1}b_1 , \, r_2^{k_2} \equiv_{a_2} b_2}}  w &\equiv_{a_1 a_2} a_2 r_1 \:\wedge\: z \equiv_{a_2} r_2 \label{eq:sim_sys_5} \, .
\end{align}
\textbf{Case (3a)}: If $k_2 \geq 3$ then \eqref{eq:sim_sys_4} satisfies the conditions of Lem.~\ref{lem:higher-k-finite-sol} (note that $a_2b_1-a_1b_2 \neq 0$ by non-similarity) and so there are finitely many effectively computable solutions, which may be checked against the modular constraints. 

\textbf{Case (3b)}: Suppose $k_1 = k_2 = 2$. If $a_1 a_2^{k_1-1}$ is a perfect square, then the left-hand side of \eqref{eq:sim_sys_4} may be factored using the difference of two squares, and by considering prime factorisations there are finitely many solutions. Otherwise, if $a_1 a_2^{k_1-1}$ is not a perfect square, \eqref{eq:sim_sys_4} comprises a generalised Pell equation, and by Cor.~\ref{cor:pell-sol-set} the solutions $(w,z)$ are exactly the value sets of finitely many pairs of simple reversible LRBS. 

Since reversible LRBS are periodic modulo $N$ for any $N \ge 1$, given an LRBS $\bm u$ one can effectively find an integer $M$ such that each subsequence $\langle u_{Mn+s} \rangle_{n=-\infty}^{\infty}$ is constant modulo $a_1a_2$ for each $0 \leq s \leq M-1$. Therefore the solution set of pairs $(w,z)$ satisfying \eqref{eq:sim_sys_4} and \eqref{eq:sim_sys_5} is comprised of a finite union of such subsequences (or possibly the empty set if none of the modular constraints are satisfied).

For every solution pair $(w,z)$, we recover a solution $x$ to the original system  \eqref{eq:sim_sys_1} by $x = \frac{z^{k_2}-b_2}{a_2}$. Since for any simple LRBS $\bm u$, we have $\left\langle \frac{u_n^{k_2}-b_2}{a_2} \right\rangle_{n=-\infty}^\infty$ is also a simple LRBS, and moreover for any $z$ satisfying the modular constraints \eqref{eq:sim_sys_5} we have $\frac{z^{k_2}-b_2}{a_2}$ is an integer, the set $S$ of solutions $x$ is a union of finitely many simple reversible LRBS\@.

\textbf{Case (4)}: Assume $l \geq 3$. If $\max\{k_1,k_2,k_3\} \geq 3$, then $S$ is finite and effectively computable by Case (3). Otherwise, $k_1 = k_2 = k_3=2$, and by the same process as in Case (3) we obtain a system
\begin{align}
    a_2y_1^2 - a_1y_2^2 &= a_2b_1-a_1b_2 \label{eq:sim_sys_6} \\
    a_3y_1^2 - a_1 y_3^2 &= a_3b_1 - a_1 b_3  \label{eq:sim_sys_7}
\end{align}
along with some modular constraints which we omit. By multiplying \eqref{eq:sim_sys_6} by $a_2a_3^2$ and \eqref{eq:sim_sys_7} by $a_3a_2^2$, and setting $w \coloneq a_2a_3y_1$, $z_2 \coloneq y_2$, $z_3 \coloneq y_3$, we obtain the simultaneous equations
\begin{align}
    w^2 - a_1a_2a_3^2 z_2^2 = a_2a_3^2(a_2b_1-a_1b_2)  \label{eq:sim_sys_8}\\
    w^2 - a_1a_2^2a_3 z_3^2 = a_2^2a_3(a_3b_1-a_1b_3) \label{eq:sim_sys_9} \, .
\end{align}
If either $a_1a_2a_3^2$ or $a_1a_2^2a_3$ is a perfect square then there are finitely many solutions by factoring the left-hand side of \eqref{eq:sim_sys_8} or \eqref{eq:sim_sys_9} using the difference of two squares. Otherwise, noting that $a_2b_1-a_1b_2$ and $a_3b_1-a_1b_3$ are non-zero by non-similarity, we have a system of simultaneous Pell equations in $w,z_2,z_3$, which has finitely many effectively computable solutions by Lem.~\ref{lem:simul-pell-finite}, so $S$ is finite and effectively computable.
\end{proof}

Prop.~\ref{prop:pos_pow_sol_sets} already gives an algorithm to decide the satisfiability of any set of constraints of the form $x > c$ and positive constraints $\mathbf{Z}^k(ax+b)$. We now show that when the solution set is infinite, negative constraints can only be violated on a subset of relative density $0$ within the solution set, meaning that arbitrarily many negative constraints will still leave infinitely many solutions overall. 
\begin{definition}
Let $\varnothing \neq T \subseteq \mathbb N$ and $S \subseteq T$. Define the \emph{(upper) density of $S$ inside $T$} to be
\begin{equation*}
    \limsup_{n \to \infty} \frac{\left|S \cap [0,n] \right|}{\left| T \cap [0,n] \right|} \, . 
\end{equation*}
\end{definition}
First we need an elementary lemma. Given a function $f : \mathbb Z \to \mathbb Z$, define $S_f(c,n) = \mathrm{Im}(f) \cap [c,n]$. 
\begin{lemma} \label{lem:zero_density}
Suppose $f \in \mathbb{Z}[x]$ has degree $d$ and a positive leading coefficient, $c_1,c_2, N > 0$, and $g : \mathbb N \to \mathbb N$ is a function such that for all $y \ge N$ we have $|g(y)| > c_2y^{d+1}$.
Then we have
\begin{equation*}
    \limsup_{n \to \infty} \frac{|S_g(c_1,n) \cap S_f(c_1,n)|}{|S_f(c_1,n)|} = 0 \, .
\end{equation*}
\end{lemma}
\begin{proof}
It is sufficient to show that
\begin{equation} \label{eq:lim_sup_0}
    \limsup_{n \to \infty} \frac{|S_g(0,n)|}{|S_f(0,n)|} = 0 \, .
\end{equation}
Since the leading coefficient of $f$ is positive, there are constants $A_1,A_2>0$ such that for all $y \geq A_1$, we have $f$ is injective and $0\leq f(y) \leq A_2 y^d$. 
Therefore,
\begin{align*}
    |S_f(0, A_2 y^d)| \ge y-A_1 \, ,
\end{align*}
and so we have
\begin{align*}
    |S_f(0,n)| \ge \left(\frac{n}{A_2} \right)^{1/d} - A_1 \, ,
\end{align*}
while $|S_g(0, n)| \le N + (n/c_2)^{1/(d+1)}$.
Then \eqref{eq:lim_sup_0} follows from the following inequality, which holds whenever $|S_f(0,n)|>0$ and $(n/A_2)^{1/d} - A_1 > 0$;
\begin{equation*} \frac{|S_g(0,n)|}{|S_f(0,n)|} \le \frac{N + (n/c_2)^{1/(d+1)}}{(n/A_2)^{1/d} - A_1}\,.\qedhere
\end{equation*}
\end{proof}


\begin{proposition} \label{prop:negative_constraints}
Let $S$ be the set of solutions $x$ to a system of constraints given by $x > c$ and non-similar constraints $\mathbf{Z}^{k_i}(a_ix+b_i)$ for $i =1 ,\dots, l$, $a_i > 0$. If $S$ is infinite, then the subset $S' \subseteq S$ for which any non-redundant negative constraint $\neg \mathbf{Z}^k(ax+b)$ is violated has density $0$ relative to $S$.
\end{proposition}
\begin{proof}
A negative constraint $\neg \mathbf{Z}^k(cx+d)$ being violated is equivalent to the positive constraint $\mathbf{Z}^k(cx+d)$ holding. 
It is sufficient to prove the result for a single negative constraint $\neg\mathbf{Z}^k(cx+d)$ as the union of finitely many null-density sets again has null density.
We go through the cases given by Prop.~\ref{prop:pos_pow_sol_sets}.

\textbf{Case (1)}: $l = 0$. Then $S = [c+1,\infty)$. The discarded set $S'$ with the constraint $\mathbf{Z}^k(ax+b)$ added becomes a finite union of sets of the form $g(\mathbb Z) \cap [c+1,\infty)$ where $g$ is a polynomial with $\deg g = k \geq 2$ and positive leading coefficient, by Prop.~\ref{prop:pos_pow_sol_sets}. Apply Lem.~\ref{lem:zero_density} with $f(y) = y$ and each $g$ to conclude $S'$ has density $0$ relative to $S$.

\textbf{Case (2)}: $l = 1$. First, suppose $\mathbf{Z}^k(ax+b)$ is similar (but not redundant) to $\mathbf{Z}^{k_1}(a_1x+b_1)$. 
By Lem.~\ref{lem:similar_coalesce}, these constraints get coalesced into $\mathbf{Z}^K(Ax+B)$ where $K$ is the least common multiple of $k$ and $k_1$. 
Then the solution set $S$ to $x > c$ and $\mathbf{Z}^{k_1}(a_1x + b_1)$ is a finite union of sets of the form $f(\mathbb Z) \cap [c+1,\infty)$ for a polynomial $f$ with positive leading coefficient and degree $k_1$, and the discarded set $S'$ to $x > c$ and $\mathbf{Z}^K(Ax+B)$ is a finite union of sets of the form $g(\mathbb Z) \cap[c+1,\infty)$ for polynomials $g$ of degree $K > k_1$. 
Therefore we may again apply Lem.~\ref{lem:zero_density} with $f,g$ to conclude $S'$ has density $0$ relative to $S$.

Otherwise, suppose $\mathbf{Z}^k(ax+b)$ is not similar to $\mathbf{Z}^{k_1}(a_1x+b_1)$. Then by Prop.~\ref{prop:pos_pow_sol_sets}, if the discarded set $S'$ to $x>c,\, \mathbf{Z}^{k_1}(a_1x+b_1),\, \mathbf{Z}^{k}(ax+b)$ is infinite, then $k = k_1 = 2$ and $S'$ is a union of finitely many simple reversible LRBS given by \cref{cor:pell-sol-set}, restricted to $[c+1,\infty)$. 
For such simple reversible LRBS $\langle u_n \rangle_{n= -\infty}^\infty$, we have that $|u_n|$ grows exponentially as $|n| \to \infty$. Therefore we may apply Lem.~\ref{lem:zero_density} to  $f$ and the function $g(y) = u_y$ for each LRBS $\bm u$ forming part of the discarded set of $S'$, to get that $S'$ has density $0$ relative to $S$.

\textbf{Case (3)}: $l = 2$. The only way in which $S$ is infinite is if $k_1=k_2=2$. In that case, it is impossible for $S'$ to be infinite. Indeed, by Prop.~\ref{prop:pos_pow_sol_sets} we have that $S'$ is infinite only if there are at most two non-similar constraints among $\mathbf{Z}^{k_1}(a_1x+b_1),\, \mathbf{Z}^{k_2}(a_2x+b_2),\, \mathbf{Z}^k(ax+b)$, meaning that (without loss of generality) $\mathbf{Z}^k(ax+b)$ is similar to $\mathbf{Z}^{k_1}(a_1x+b_1)$. By Lem.~\ref{lem:similar_coalesce} we can coalesce these constraints into $\mathbf{Z}^K(Ax+B)$, where $K$ is the least common multiple of $k,k_1$. But by non-redundancy, we cannot have $k \mid k_1$ nor $k_1 \mid k$ so $K > k_1 = 2$. Therefore by Prop.~\ref{prop:pos_pow_sol_sets} $S'$ is finite, and so trivially has density $0$ relative to $S$. 
\end{proof}

Prop.~\ref{prop:negative_constraints} was the final step in our proof of Thm.~\ref{thm:1-pres-Zk}, to the effect that $\mathit{FO}^1 \langle   \mathbb{Z}; 0, 1, +, -, <, (\equiv_m)_{m>1}, (\mathbf{Z}^k)_{k > 1}\rangle$ is decidable. Indeed, in summary, we use Prop.~\ref{prop:pre-1-pres-Zk} to reduce to considering satisfiability of constraints of the form $x>c,\mathbf{Z}^k(ax+b),\neg \mathbf{Z}^k(ax+b)$. We may further reduce to the case in which the positive constraints $\mathbf{Z}^k(ax+b)$ are all non-similar, and all power constraints are non-redundant with respect to each other. Prop.~\ref{prop:pos_pow_sol_sets} shows that the solution set to any number of positive constraints $\mathbf{Z}^k(ax+b)$ is effectively computable, and in the case for which the solution set is infinite, Prop.~\ref{prop:negative_constraints} shows that the addition of any negative constraints $\neg \mathbf{Z}^k(ax+b)$ removes at most a subset of null density, so the solution set remains infinite. Meanwhile if the solution set is finite, one solves the decision problem by simply enumerating every solution and checking against all constraints. 

\section{Quadratic and Cubic Predicates}
\label{sec4}
In this section, we adapt the techniques used to prove Thm.~\ref{thm:1-pres-Zk} to decide single-variable Presburger arithmetic expanded with multiple predicates $(\mathcal{R}_i)_i$, where each predicate $\mathcal{R}_i$ corresponds to the value set of an integer-valued univariate polynomial $f_i$ of degree at most~$3$, i.e., $\mathcal{R}_i(x)$ holds if and only if there exists an integer $u$ such that $f_i(u) = x$. Formally, we prove the following.

\begin{theorem}
\label{thm:2-pres-sqcu}
    Let $(\mathcal{R}_i)_i$ be predicates corresponding to value sets of integer-valued polynomials $(f_i)_i$ of degree at most 3.
    Then the theory $\mathit{FO}^1 \langle   \mathbb{Z}; 0, 1, +, -, <, (\equiv_m)_{m \geq 2}, (\mathcal{R}_i)_{i}\rangle$ is decidable.
\end{theorem}

Before proceeding with the technical proof, we record a few simplifying assumptions. These assumptions establish an analogue of Prop.~\ref{prop:pre-1-pres-Zk} (i.e., the pre-processing step) \emph{mutatis mutandis}. 
Moreover, a polynomial $f_i$ of degree $0$ is constant and so $\mathcal{R}_i(ax+b)$ is equivalent to $ax + b = f_i(0)$, and for a polynomial $f_i(t) = dt+e$ of degree $1$, $\mathcal{R}_i(ax+b)$ is equivalent to $ax+b \equiv e \pmod{d}$.
Hence we can assume that the polynomials $f_i$ are of degree 2 or 3.
In summary, we obtain the following.
\begin{lemma}
    The decision problem in Thm.~\ref{thm:2-pres-sqcu} Turing-reduces to deciding whether there exists $x > 0$ satisfying a conjunction of exactly one constraint $x > c$ together with other constraints of the form $\mathcal{R}_i(ax+b)$ and $\lnot \mathcal{R}_i(ax+b)$, where $\mathcal{R}_i = f_i(\mathbb{Z})$ for a polynomial $f_i$ of degree two or three. 
\end{lemma}

Next we want to restrict the kinds of polynomials that can appear. 
Recall that a polynomial $f(t) = c_dt^d + c_{d-1}t^{d-1} + \cdots + c_0$ is \emph{depressed} if its second-highest coefficient, $c_{d-1}$, is zero.

Let $\mathcal{R}(ax+b, q, r)$ denote the predicate $\{f(q u + r) \mid u \in \mathbb{Z}\}$. We add this ternary predicate to our signature and henceforth always assume that all our polynomials are depressed and monic.

\begin{lemma}
\label{lem:depressed}
     For any integer-valued polynomial $f$ of degree at most $3$ and corresponding predicate $\mathcal{R}$ and constants $a, b \in \mathbb{Z}$, we can compute a depressed monic polynomial $\tilde f$ with integer coefficients, together with constants $\tilde a, \tilde b, q, r$, such that for all $x$, $\mathcal{R}(ax + b)$ is equivalent to $\tilde{\mathcal{R}}(\tilde ax + \tilde b, q, r)$.
\end{lemma}
\begin{proof}
    We tackle the degree-$3$ case, the degree-$2$ case being similar and simpler. 

    The predicate $\mathcal{R}(ax+b)$ is equivalent to $\exists u \, .\, c_3 u^3 + c_2 u^2 + c_1 u + c_0 = ax + b$. We can multiply through by an appropriate integer and assume without loss of generality that $c_3, \ldots, c_0$ are integers, and $c_3$ is positive. We ``complete the cube'' by multiplying through by $27c_3^2$ and write the equivalent statement
    $$
    \exists u \, .\, (3c_3u+c_2)^3 + 27c_1c_3^2 u + 27c_0c_3^2 - 9c_2^2c_3 u - c_2^3 = 27c_3^2(ax+ b),
    $$
    which can be further rearranged as
    \begin{align*}
    &\exists u \, .\, (3c_3u+c_2)^3 + (9c_1c_3 - 3c_2^2)(3c_3u + c_2) \\
    &=~ 27ac_3^2x + (27bc_3^2 - 27c_0c_3^2 + 9c_1c_2c_3 - 2c_2^3) \, ,
    \end{align*}
    or in other words
    $\exists u \, .\, \tilde f(3c_3 u + c_2) = \tilde a x + \tilde b$. 
\end{proof}

We continue following a very similar strategy to find a witness $x$ that satisfies all constraints as in Sec.~\ref{sec:power predicates}.
Each predicate $\mathcal{R}_i$ can occur both positively and negatively, and we want to show that we can enumerate a solution set satisfying positive constraints when it is finite; and when it is infinite, adding a further ``non-similar'' positive constraint would in all but one case restrict the solution set to a subset of relative null density. In the exceptional case, the solutions to the positive constraints are parametrised as an LRS, and the discarded indices form arithmetic progressions. Thus, in this case too, we can effectively determine whether there remains a value of $x$ not discarded by the negative constraints. 

As we did previously, we need to account for redundant constraints, but it is not immediately clear what a meaningful definition of redundancy is.
Recall that a multivariate polynomial is \emph{absolutely irreducible} if it is irreducible over the complex numbers.

\begin{definition}
The constraint $\mathcal{R}_1(a_1x+b_1)$ 
is \emph{redundant with respect to} $\mathcal{R}_2(a_2x+b_2)$ if $ \deg(f_1) \mid  \deg(f_2)$ and $a_2 f_1(u_1) - a_1 f_2(u_2) - a_2 b_1 + a_1 b_2 \in \Q[u_1,u_2]$ is not absolutely irreducible.
\end{definition}


As the polynomials defining our predicates have degrees 2 or 3, two constraints can only be in a redundancy relationship if the underlying polynomials have the same degree. We now have:

\begin{proposition}
    Let $\mathcal{R}_1(a_1x+b_1)$ be redundant with respect to $\mathcal{R}_2(a_2x+b_2)$. 
    Then $a_1b_2 = a_2b_1$.
\end{proposition}
\begin{proof}
Recall from Lem.~\ref{lem:depressed}, we assume that $f_1$ and $f_2$ are depressed and satisfy $f_i(0) = 0$.
Thus, when $\deg(f_1)=\deg(f_2) = 2$, write $f_i(u_i) = c_i u_i^2$ for $i = 1,2$.
Then we have, by the definition of absolutely reducibility: $$A_1 u_1^2  - A_2 u_2^2  + a_1b_2 -a_2b_1 = A_1(u_1 + C_2u_2 + D)(u_1 + C_2'u_2 + D')\,,$$ where $A_1 = a_2c_1$ and $A_2 = a_1c_2$ are non-zero. 
Then, $A_1C_2C_2' = -A_2$, $C_2+C_2' = 0$, $D + D' = 0$, $DC_2'+D'C_2 = 0$, and $A_1DD' = a_1b_2 -a_2b_1$.
Thus the second and third equations imply that $D = -D'$ and $C_2 = -C_2'$, and the first implies that $C_2 \ne 0$ as $A_2 \ne 0$. Hence the fourth equation implies that $D = 0$ and so the last equation yields $a_1b_2 -a_2b_1 = 0$. 

When $\deg(f_1)=\deg(f_2) = 3$, let $f_i(u_i) = c_i u_i^3 + d_iu_i$ for $i = 1,2$.
Then we have by absolute reducibility:
\begin{align*}
    &\quad A_1 u_1^3  + B_1u_1 - A_2 u_2^3  - B_2u_2 + a_1b_2 -a_2b_1 \\
    &= A_1(u_1 + C_2u_2 + D)(u_1^2 + E_1u_1u_2 + C_2'u_2^2 + E_2u_1 +E_3u_2 + D')\,,
\end{align*}
where $A_1 = a_2c_1$ and $A_2 = a_1c_2$ are non-zero, $B_1 = a_2d_1$ and $B_2 = a_1d_2$. 
Then, $A_1C_2C_2' = -A_2$ and $A_1DD' = a_1b_2 -a_2b_1$ and $E_1 + C_2$, $E_2 + D $, $C_2' + E_1C_2$, $E_3+C_2E_2+E_1D$, $C_2 E_3+DC_2'$ are all zero.
Substituting $E_1 = -C_2$ and $E_2 = -D$ gives that $C_2' - C_2^2$, $E_3-2C_2D$, and $C_2 E_3+DC_2'$ are all zero. 
Hence, as $C_2 \ne 0$ as $A_2 \ne 0$, $E_3=C_2D$, which forces that $C_2D = 0$ and thus that $D = 0$. 
Thus $a_1b_2 -a_2b_1 = 0$. 
\end{proof}

In the case of redundancy with cubic polynomials, we can thus solve the (Diophantine) equation $a_2 f_1(u_1) - a_1 f_2(u_2) = 0$ via the factorisation $A_1u_1^3 + B_1u_1 - A_2u_2^3 - B_2u_2  = A_1(u_1 + mu_2)(u_1^2 - m u_1u_2 + m^2u_2^2 + B_1/A_1)$, where $m = -B_2/B_1 = -(A_2/A_1)^{1/3}$ is non-zero.  This follows using the notation in the proof above.
We shall assume that $m$ is rational. This is of course the case when $B_2 \ne 0$, or when $B_1 = B_2 = 0$ and $A_2/A_1$ is a perfect cube. Otherwise, it is impossible for $A_1x^3 - A_2y^3 = 0$ to have integer solutions.

Looking at the proof of the lemma above, when two redundant predicates are both satisfied, a linear relationship between $u_1$ and $u_2$ has to hold, or in the cubic case, $u_1$ and $u_2$ have to lie on a certain conic. 
This conic is $u_1^2 - m u_1u_2 + m^2u_2^2 + B_1/A_1 = 0$, which represents an ellipse, which thus contains finitely many integer points $(u_1, u_2)$ that we can effectively compute. 
Hence we conclude the following.
\begin{lemma}
    Let $\phi(x)$ be a conjunction of two redundant predicates $\mathcal{R}_1(a_1x+b_1, q_1, r_1)$ and either $\mathcal{R}_2(a_2x+b_2, q_2, r_2)$ or $\lnot\mathcal{R}_2(a_2x+b_2, q_2, r_2)$.
    Then $\phi(x)$ can be written as the conjunction of one predicate $\mathcal{R}_3(a_3x+b_3, q_3, r_3)$ and a finite number of atoms definable in quantifier-free Presburger arithmetic.
\end{lemma}

As we can observe from the factorisation, the solution set to $a_2 f_1(u_1) - a_1 f_2(u_2) = 0$ is the union of integer points on a line (passing through the origin and having rational slope $m = -s/t$), and finitely many integer points on a bounded conic. In other words, a solution $u_1$ corresponds to a solution $u_2$ if and only if the former takes one of finitely many values, or satisfies a set of divisibility constraints. 
The constraint on $a_2 x + b_2$ is thus ``redundant'' in view of the constraint on $a_1 x + b_1$ in the sense that it does not add ``algebraic'' information beyond modular-arithmetic annotation. 


We can henceforth focus on the case where there is no redundancy, and at least one of the predicates corresponds to a cubic polynomial.

\begin{proposition}
 \label{prop:pos_poly_sol_sets}
The solution set of all $x \in \mathbb{Z}$ satisfying $l$ non-redundant constraints of the form $\mathcal{R}_i(a_ix+b_i, q_i, r_i)$ with $a_i>0$ and $d_i=\deg(f_i)$ for $1 \leq i \leq l$ is effectively computable and has the following structure.
\begin{enumerate}
    \item If $l = 0$ then $S = \mathbb{Z}$.
    \item If $l = 1$ then either $S = \varnothing$ or $S$ is a finite union of sets of the form $f(\mathbb{Z})$ where $f$ is a polynomial of degree $d_1$ and with a positive leading coefficient.
    \item If $l = 2$ then either $S = \varnothing$ or one of the following holds.
    \begin{enumerate}
        \item If $d_1=d_2= 3$ then $S$ is finite.
        \item If $d_1=d_2=2$ then $S$ is a union of finitely many simple reversible LRBS. 
        \item If $d_1\neq d_2$ then either $S$ is finite or $S$ is a finite union of sets of the form $f(\mathbb{Z})$ where $f$ is a polynomial of degree~$6$.
        
    \end{enumerate}
    \item If $l= 3$ then either $S = \varnothing$ or one of the following holds.
    \begin{enumerate}
        \item If any two $i\neq j$ exist such that $d_i=3= d_j$, then $S$ is finite.
        \item If $d_1=d_2=d_3=2$ then $S$ is finite.
        \item Otherwise $S$ is a union of finitely many simple reversible LRBS. 
    \end{enumerate}
    \item If $l\geq 4$ then $S$ is finite.
\end{enumerate}
\end{proposition}
\begin{proof}
\textbf{Case (1)} is obvious.

\textbf{Case (2)}:
Write $f(u)$ for $f_1(q_1 u + r_1)$.
If $l = 1$ then $S = \varnothing$ if $b_1$ does not lie in the value set of $f$ modulo $a_1$, and this can be checked algorithmically by enumerating all values of $f$ modulo $a_1$. 
Conversely, if $b_1$ does lie in the value set of $f$ modulo $a_1$ then there are infinitely many solutions and they can be parametrised as follows. Pick $u \in \mathbb{N}$ such that $f(u) \equiv b_1 \bmod a_1$. Then consider the polynomial $g_{u,a_1} \in \mathbb{Z}[t]$ that satisfies
\begin{equation*}
    f(u+ta_1) =  a_1g_{u,a_1}(t) + b_1 \, .
\end{equation*}
Then $g_{u,a_1}$ has degree $d_1$ and a positive leading coefficient (as $a_1, q_1 > 0$). Every $t \in \mathbb{Z}$ gives rise to a solution $x = g_{u,a_1}(t)$ of $\mathcal{R}_1(a_1x+b_1, q_1,r_1)$. Conversely, whenever $x \in \mathbb{Z}$ is a solution, then there must exist some $y = u+ta_1$ such that $f(y)=a_1x+b_1$, $f(u) \equiv b_1 \bmod a_1$, and $t \in \mathbb{Z}$. Therefore the solution set is exactly 
\begin{equation*}
    S \, \, \, = \!\!\!\!\!\!\!\! \bigcup_{\substack{0\leq u < a_1 \\ f(u) \equiv b_1 \bmod a_1}} \!\!\!\!\!\!\!\!\!\! g_{u,a_1}(\mathbb Z) \, .
\end{equation*}


\textbf{Case (3a)}
We claim that a pair of non-redundant cubic positive constraints $\mathcal{R}_1(a_1 x + b_1, q_1, r_1), \mathcal{R}_2(a_2 x + b_2, q_2, r_2)$ has finitely many solutions, which can moreover be effectively enumerated. Indeed, homogenising the cubic curve $$a_2 f_1(u_1) - a_1f_2(u_2) - a_2 b_1 + a_1 b_2 = 0$$ into the projective plane over an appropriate algebraic extension of the rationals gives an absolutely irreducible curve. This curve has three places at infinity, $([\rho_i: 1: 0])_{i=1}^3$ where each $\rho_i$ is a complex cube root of $a_2/a_1$ (because $f_1$ and $f_2$ are monic, evaluating this curve at $[u_1:u_2:0]$ gives $a_2u_1^3 - a_1u_2^3 = 0$). If this curve has genus $1$, then its finitely many integer points can be enumerated by \cite[Thm.~4.3]{Baker_1975}. Otherwise, the curve has genus $0$, in which case its finitely many integer points can be enumerated by \cite{poulakis-three-place}.

\textbf{Case (3b)} This case follows from Case (3b) of Prop.~\ref{prop:pos_pow_sol_sets} \emph{mutatis mutandis}.

\textbf{Case (3c)}
We now consider the case of two positive constraints, where the first $\mathcal{R}_1(a_1 x + b_1, q_1, r_1)$ corresponds to a quadratic polynomial. The attendant curve $a_1 f_2(u_2) - a_2f_1(u_1) = a_1b_2 - a_2 b_1$ can then be simplified to have the form $(a_2 u_1)^2 = a_2(a_1 f_2(u_2) + a_2 b_1 - a_1 b_2) = a_1 a_2 g(u_2)$, where $g \in \mathbb{Q}[u_2]$ is monic, and $a_1 g \in \mathbb{Z}[u_2]$. If $g$ has three distinct roots (i.e., an elliptic curve has arisen), then the finitely many integer points on the curve can be enumerated by \cite[Thm.~2]{Baker1969hyperelliptic}.

Otherwise, $g$ has a repeated root, which is necessarily rational because it corresponds to the common factor of $g$ and its derivative. By Gauss's lemma, $a_1 g$ thus splits over $\mathbb{Z}$ as $(\alpha u_2 + \beta)^2(\gamma u_2 + \delta)$. We then make the substitution $v_1 = \frac{a_2 u_1}{\alpha u_2 + \beta}$, and observe that $v_1^2 = a_2(\gamma u_2 + \delta)$. In this manner, the values of $u_1$ and $u_2$ are parametrised by $v_1$. If we have constraints that $u_1$ and $u_2$ are respectively $r_1$ modulo $q_1$ and $r_2$ modulo $q_2$, we enforce the following modular-arithmetic constraints on $v_1$:
\begin{align*}
    v_1^2 \equiv a_2 r_2 \gamma + a_2 \delta &\mod a_2 \gamma q_2, \\
    \alpha v_1^3 + (a_2 \beta \gamma - a_2 \alpha \delta)v_1 \equiv a_2^2r_1\gamma &\mod a_2^2  \gamma q_1. 
\end{align*}
Inserting $v_1^2 = a_2(\gamma u_2 + \delta)$ into $(\alpha u_2 + \beta)^2(\gamma u_2 + \delta)$ we obtain that $x$ is parametrised as a degree $6$ polynomial in $v_1$.

We thus get solutions to our constraints whenever $v_1$ satisfies the above: such values for $v_1$, if they exist, form a union of finitely many arithmetic progressions by the Chinese Remainder Theorem. 

\textbf{Case (4a)} immediately follows from Case (3a).

\textbf{Case (4b)}  follows from  Case 4 of Prop.~\ref{prop:pos_pow_sol_sets} \emph{mutatis mutandis}.

\textbf{Case(4c)}
We freely borrow notation from Case (3c). 
Assume $f_2$ has degree 3 and that $f_1$ and $f_3$ have degree 2. 
Then, as in Case (3c), we construct $v_1$ and $v_3$ that give the system 
\begin{align*}
    v_1^2 &= a_2 (\gamma_1 u_2 + \delta_1) \, , \\
    v_3^2 &= a_2 (\gamma_3 u_2 + \delta_3) \, ,
\end{align*}
where as before, $v_i = \frac{a_2 u_i}{\alpha_i u_2 + \beta_i}$. Eliminating $u_2$ yields $\gamma_3 v_1^2 - \gamma_1 v_3^2 = a_2(\gamma_3 \delta_1 - \gamma_1 \delta_3)$.

We argue that the constant on the right is non-zero. Suppose for the sake of deriving a contradiction that $\gamma_3\delta_1 - \gamma_1\delta_3 = 0$. This would imply that the polynomials $a_1 a_3 g_1 = a_3(a_1 f_2 + a_2 b_1 - a_1 b_2)$ and $a_1 a_3 g_3 = a_1(a_3 f_2 + a_2 b_3 - a_3 b_2)$ have the root $-\delta_1/\gamma_1$ in common. This would also have to be a root of their difference $a_2(a_3b_1 - a_1b_3)$, which by our assumption of non-redundancy, is a non-zero constant: a contradiction, as desired.

The equation $\gamma_3 v_1^2 - \gamma_1 v_3^2 = a_2(\gamma_3 \delta_1 - \gamma_1 \delta_3)$ thus has infinitely many solutions only if it is a generalised Pell equation. In this case, by Cor.~\ref{cor:pell-sol-set} the values of $v_1$ for which all three constraints are satisfied form a union of finitely many exponentially-growing LRBS. 

\textbf{Case (5)}:
Four non-redundant (positive) constraints will have finitely many solutions by virtue of containing two cubic constraints (Case (3a)) or three quadratic constraints (Case (4b)) and may be effectively enumerated.
\end{proof} 
\begin{proposition} \label{prop:negative_constraints_poly}
Let $S$ be the set of solutions $x$ to a system of constraints given by $x > c$ and non-redundant constraints $\mathcal{R}_i(a_ix+b_i, q_i, r_i)$ for $i =1 ,\dots, l$, $a_i > 0$. If $S$ is infinite, then the subset $S' \subseteq S$ for which any non-redundant negative constraint $\neg \mathcal{R}(ax+b)$ is violated has density $0$ relative to $S$, unless $S$ is of the form $(3b)$ and $S'$ is of the form $(4c)$ in Prop.~\ref{prop:pos_poly_sol_sets}.
\end{proposition}
\begin{proof}
A negative constraint $\neg\mathcal{R}_i(ax+b)$ being violated is equivalent to the positive constraint $\mathcal{R}_i(ax+b)$ holding. 
We go through the cases given by Prop.~\ref{prop:pos_poly_sol_sets}. In each case, we show that a negative constraint discards either a subset of relative density $0$, or that the set of parameters (of the solutions to the positive constraints) invalidated by forming finitely many arithmetic progressions. In either case, we can effectively determine whether there remain solutions after accounting for finitely many negative constraints.

\textbf{Case (1)}: $l = 0$. Then $S = [c+1,\infty)$. The discarded set $S'$ with the constraint $\mathcal{R}_i(ax+b)$ added becomes a finite union of sets of the form $g(\mathbb Z) \cap [c+1,\infty)$ where $g$ is a polynomial with $\deg g = d \geq 2$ and positive leading coefficient, by Prop.~\ref{prop:pos_poly_sol_sets}. Apply Lem.~\ref{lem:zero_density} with $f(y) = y$ and each $g$ to conclude $S'$ has density $0$ relative to $S$.

\textbf{Case (2)}: $l = 1$.
Then the solution set $S$ to $x > c$ and $\mathcal{R}_1(a_1x + b_1)$ is a finite union of sets of the form $f(\mathbb Z) \cap [c+1,\infty)$ for a polynomial $f$ with positive leading coefficient and degree $d_1$. Now,
$S'$ will satisfy the conclusion of either Case (3a), (3b), or (3c). 
In case $S'$ satisfies (3a) it is finite and thus has density $0$ relative to $S$.
The case of $S'$ satisfying (3b) has been handled in Prop.~\ref{prop:negative_constraints} \emph{mutatis mutandis} and $S'$ thus has density $0$ relative to $S$.
And in case $S'$ satisfies (3c), the discarded set $S'$ is a finite union of sets of the form $g(\mathbb Z) \cap[c+1,\infty)$ for polynomials $g$ of degree $6$, which is greater than $d_1$. 
Therefore we may again apply Lem.~\ref{lem:zero_density} with each pair $f,g$ to conclude $S'$ has density $0$ relative to $S$.

\textbf{Case (3)}: $l = 2$.  $S$ is infinite only if $d_1= 2, d_2 = 3$, or $d_1 = d_2 = 2$.
In the former case, $S'$ can be infinite only if $d_3 = 2$. Then by Prop.~\ref{prop:pos_poly_sol_sets}, if the discarded set $S'$  is infinite, then $S'$ is a union of finitely many simple reversible LRBS restricted to $[c+1,\infty)$. As before we conclude by Lem.~\ref{lem:zero_density}  that $S'$ is a null-density subset of $S$. 
The case $d_1=d_2=2, d_3 = 3$ results in finitely many solutions being discarded if a pair of constraints gives rise to an elliptic curve (see the first part of Case (3c) of Prop.~\ref{prop:pos_poly_sol_sets}); the case where it does not result in infinite LRBS of solutions being discarded. This discarded set can have positive relative density; however we have from Case (3c) of Prop.~\ref{prop:pos_poly_sol_sets} that the solutions to the positive constraints are themselves parametrised as LRBS\@. Prop.~\ref{prop:scarycase} shows that the indices of discarded solutions form arithmetic progressions, and hence we can effectively determine whether there exists a solution not discarded by the negative constraints.
 
\textbf{Case (4)} 
 $S'$ is finite and hence has density $0$ relative to $S$.
\end{proof}

The following is a special case of \cite[Prop.~2]{Schlickewei1993}.
\begin{lemma}[Parametrisation of $x$ values ruled out]
\label{lem:neg}
    Let $D_1,D_2 > 0$ be square-free integers. Let $(\alpha_n, \beta_n)$ be a pair of LRBS of solutions to the generalised Pell equation $\alpha^{2}-D_1 \beta^{2}=N_1$, and let $(\sigma_n,\tau_n)$ be a pair of LRBS of solutions to $\sigma^{2}-D_2\tau^{2}=N_2$, given by Cor.~\ref{cor:pell-sol-set}. Let $\phi_1,\phi_2$ be rational polynomials and $\mathcal K$ be the set of $k \in \mathbb Z$ such that there exists $m \in \mathbb Z$ with \begin{equation}
    \label{eqn:pellmerge}
    (\alpha_k, \beta_k)=(\phi_1(\sigma_{m}), \phi_{2}(\tau_{m}))\,.
    \end{equation}
    Suppose $\mathcal K$ is infinite, then we have
    \begin{itemize}
        \item[(i)] $D_1=D_2$;
        \item[(ii)] $\deg \phi_1=\deg \phi_2=:r$; and 
        \item[(iii)] The solutions $(\alpha_k,\beta_k)$ to \eqref{eqn:pellmerge} can be parametrised; there exists effectively computable $c \in \mathbb Z$ such that $\left\{(\alpha_k,\beta_k) : k \in \mathcal K \right\} = \left\{(\alpha_{rm + c},\beta_{rm + c}) : m \in \mathbb{Z} \right\}$.
    \end{itemize}
\end{lemma}

\begin{proof}
It is clear from the polynomial relations that the field extensions $\mathbb{Q}(\sqrt{D_1})$ and $\mathbb{Q}(\sqrt{D_2})$ are equal, which can happen only if $D_1=D_2$.

By Cor.~\ref{cor:pell-sol-set}, we have
\begin{equation*}
(\alpha_k, \beta_k)=(A_1\varepsilon^{k}+A_2\varepsilon^{-k}, B_1\varepsilon^{k}+B_2\varepsilon^{-k})
\end{equation*}
and
\begin{equation*}
(\sigma_{m}, \tau_{m})=(A_3\varepsilon^{m}+A_4\varepsilon^{-m}, B_3\varepsilon^{m}+B_4\varepsilon^{-m}) \, ,    
\end{equation*}
where $\varepsilon$ is the fundamental unit associated with the underlying Pell equation, and $A_1,\dots A_4$, $B_1,\dots, B_4$ are algebraic numbers.

Substituting the latter equation into (\ref{eqn:pellmerge}) and comparing growth rates, one sees that \eqref{eqn:pellmerge} holding for infinitely many $k \in \mathbb N$ forces $\deg \phi_1=\deg \phi_2$ (which we call $r$). Further, for infinitely many pairs $(k,m)$ satisfying \eqref{eqn:pellmerge}, one has $k=r m +c$, for a constant integer $c$ not depending on $k$ or $m$. This can be seen by dividing out by $\varepsilon^{r m }$ after making the  substitution in (\ref{eqn:pellmerge}), and taking the limit as $ k,m\rightarrow\infty$. In particular, we must have $\frac{\varepsilon^{k}}{\varepsilon^{r m}}$ converges to a constant as $k, m\rightarrow \infty$. As $\varepsilon>1$, we have that this sequence becomes eventually constant, i.e., $\varepsilon^{k}/\varepsilon^{r m}=\varepsilon^{c}$ for some fixed integer $c$. We now claim that we must have $\phi_1(\sigma_m) = \alpha_{rm+c}$ for every $m\in \mathbb{Z}$. Indeed, one can consider the polynomial in $x$ of degree $2r$ given by
$\Lambda(x)=x^{r}\left(\phi_1(A_3 x +A_4x^{-1})- A_1\varepsilon^{c}x^{r}-A_2\varepsilon^{-c}x^{-r}\right)$.

We have just established that $\phi_1(\sigma_m) = \alpha_{r m + c}$ for infinitely many $m$, therefore $\Lambda(x)$ is a polynomial in $x$ with infinitely many roots of the form $x=\varepsilon^m$, for infinitely many values of $m$, hence must be identically zero. This shows that $\phi_1(\sigma_m) = \alpha_{rm+c}$ for all $m$, and by the same argument, we have $\phi_2(\tau_m) = \beta_{rm+c}$, from which (iii) follows.

Furthermore, this constant integer $c$ can be effectively determined from the constants in the equations. 
\end{proof}

\begin{proposition}\label{prop:scarycase}
    Suppose $S$ is of the form (3b) and $S'$ is of the form (4c) of Prop.~\ref{prop:pos_poly_sol_sets}. Then $S'$ consists of indices of $S$ in arithmetic progressions. These may be computed effectively.  
\end{proposition}
\begin{proof}
We borrow notation from the relevant cases. Note that $S$ is parametrised by
$u_1,u_3$ satisfying the generalised Pell equation
$\gamma_3 u_1^2 - \gamma_1u_3^2 = C$ for some $C$ depending on the data,
and that further $u_1^2 =(\alpha_1 u_2 + \beta_1)^2(\gamma_1 u_2 + \delta_1)$
and also $u_3^2 =(\alpha_3 u_2 + \beta_3)^2(\gamma_3 u_2 + \delta_3).$
Note that $S'$ is parametrised by 
$v_1,v_3$ satisfying  $\gamma_3 v_1^2 - \gamma_1 v_3^2 = a_2(\gamma_3 \delta_1 - \gamma_1 \delta_3)$
with $v_i^2 = a_2(\gamma_iu_2+\delta_i)$ for $i=1,3.$
hence $u_2 =\frac{v_i^2-a_2\delta_i}{\gamma_ia_2}$
and thus $u_i^2 =(\alpha_i \frac{v_i^2-a_2\delta_i}{\gamma_ia_2} + \beta_i)^2\frac{v_i^2}{a_2}$ for $i=1,3$ so $u_i =\pm (\alpha_i \frac{v_i^2-a_2\delta_i}{\gamma_ia_2} + \beta_i) \frac{v_i}{\sqrt{a_2}}$ for $i=1,3$. If $\sqrt{a_2}$ is irrational, the claim is vacuous; otherwise we apply Lem.~\ref{lem:neg} four times for each choice of sign and get the desired conclusion.
\end{proof}
The above results justify the correctness of the following algorithm: We first compute $S$. Should $S$ be of the form (3b) and any $n\geq 1$ negative constraint gives rise to a discarded set $S'$ of the form (4c) we compute the resulting $S\setminus \bigcup_{i=1}^n S'_i$ using Prop.~\ref{prop:scarycase}. We set $\tilde{S} =S\setminus \bigcup_{i=1}^n S'_i$ in this case. Otherwise we set $\tilde{S}=S$.  If $\tilde{S}$ is finite we enumerate $\tilde{S}$ and check all constraints. If $\tilde{S}$ is infinite we simply return true as any negative constraint will remove only a null-density subset.

\section{Undecidability}
\label{sec:hardness}
\subsection{B\"uchi's Problem}
B\"uchi formulated the following problem while studying the existential fragment of Presburger arithmetic expanded with the perfect-square predicate $\mathbf{Z}^2$: does there exist an $M$ such that any integer sequence of $M$ squares whose second difference is constant and equal to $2$ is necessarily a sequence of \emph{consecutive} squares? That is, is there an $M$ such that for all $x_1,\dots,x_M$ such that $x_{i+2}^2 - 2x_{i+1}^2 + x_i^2 = 2$ for $i = 1,\dots,M-2$, we have that $x_i = x_{i+1} -1$ for $i = 1,\dots,M-1$?
A positive answer to B\"uchi's problem enables one to define the squaring function from
the perfect-square predicate without the need for quantifiers. Indeed, the assertion $y = x^2$ would be equivalent to $\bigwedge_{i=0}^{M-1} \mathbf{Z}^2(y + 2ix + i^2)$. 
Multiplication would in turn be positive-existentially defined using the identity $4xy = (x+y)^2 - (x - y)^2$. The undecidability of the existential fragment of Presburger arithmetic expanded with the perfect-square predicate $\mathbf{Z}^2$ would hence follow.
B\"uchi himself conjectured that $M = 5$, and a proof has recently been announced by Xiao \cite{xiao2025buchi}.

We show that the (negation of the) B\"uchi conjecture can be encoded in the two-variable existential theory $\mathit{SMT}^2\langle \mathbb{Z}; +, 0, 1, \mathbf{Z}^2 \rangle$. Indeed, any counterexample sequence must have the form $i^2 + c_1 i + c_0$ for $i = 1, \ldots, 5$, where $c_0, c_1$ are integers and the polynomial $g(t) = t^2 + c_1t + c_0$ is not of the form $(x+b)^2$. By the contrapositive of \cite[Cor.~1.7]{pos-exists-mult}, there exists $M$ such that $g(M)$ is not a perfect square: by appropriate shifting, we can assume $M = 6$. The negation of the B\"uchi conjecture is then simply the formula 
$$
\exists c_0, c_1 \, . \left(\bigwedge_{i=1}^5 \mathbf{Z}^2(c_0 + ic_1 + i^2) \right) \land \neg \mathbf{Z}^2(c_0 + 6c_1 + 36) \, .
$$

The following technical but elementary lemma shows how one can encode arbitrary Diophantine equations in the signature of $\langle \mathbb{Z}; 0, 1, +, -, \mathbf{Z}^2 \rangle$ with a limited budget of first-order variables. 
\begin{lemma}
\label{lem:encode-dio}
    Let $h \in \mathbb{Z}[x_1, \ldots, x_n]$. The assertion $h(x_1, \ldots, x_n) = 0$, where $x_1, \ldots, x_n$ are integer-valued variables, can be encoded in $\mathit{FO}\langle \mathbb{Z}; 0, 1, +, -, \mathbf{Z}^2 \rangle$ via a formula that uses at most $4$ bound variables, all of which are existentially quantified.
    
\end{lemma}
\begin{proof}
    We need existentially quantified variables to implement multiplication using the identity $4xy = (x+y)^2 - (x - y)^2$, e.g., the proposition $t = 4xy$ is equivalent to $\exists u \exists v \, .\, (t = u -v) \land (u = (x + y)^2) \land (v = (x - y)^2)$, where the last conjunct is written as $\bigwedge_{i=0}^4 \mathbf{Z}^2(v + 2i(x-y)+i^2)$, and similarly for the penultimate conjunct.

    We momentarily leave aside the issue of the scarcity of existentially quantified variables and introduce rewrite rules (that replace polynomials with linear combinations of variables and simpler polynomials), with each application introducing fresh quantified variables. For convenience, we refer to the sum of monomials $r$ through $m$ of $h(x_1, \ldots, x_n)$ as $h_r$, and the $r$-th monomial $g_{r1}$ of degree $d_r$ is constructed through the intermediate monomials $g_{rs} = \frac{c_r}{4^{s-1}} \prod_{l=s}^{d_r} x_{j_l}$. In the rewrite rules that follow, the subformulas are assumed to be minimal.
    \begin{enumerate}
        \item A subformula of the form $\psi\left(h_r, \ldots\right)$ is rewritten as 
        $$
        \exists t \, .\, \psi\left[h_r \middle/ \left(t+g_{r1}\right)\right] \land \left(t = h_{r+1}\right),
        $$
        reducing the number of monomials.
        \item A subformula of the form $\psi\left(g_{rs}, \ldots \right)$ is rewritten as
        $$
        \exists u, v \, .\, \psi\left[ g_{rs} \middle/ (u - v) \right] \land ~u = \left( x_{j_s} + g_{r,s+1} \right)^2 \land~ v = \left( -x_{j_s} + g_{r,s+1} \right)^2,
        $$
        reducing the degree of the monomial.
        \item A subformula of the form $w = T^2$ is rewritten as
        $$
        \bigwedge_{i=0}^{4} \mathbf{Z}^2 \left(w + 2iT + i^2\right),
        $$
        where $T$ is a linear combination of variables.
    \end{enumerate}
    We observe that we can first repeatedly apply Rule 1 until the formula involves only monomials, then repeatedly apply Rule 2 until all propositions are either linear equations or assertions of a square relation, and finally apply Rule 3 to encode the latter in our signature. 

    Finally, we show that while applying the rewrite rules to the formula $h(x_1, \ldots, x_n) = 0$ and thus introducing existentially quantified variables, we can recycle these variables so that we only need $4$ of them. The key observation is that if a variable $t$ does not occur in a subformula $\psi$ being rewritten, it can be recycled for the purpose.

    We claim that we can alternate between introducing $t_0$ and $t_1$ as we repeatedly apply Rule 1. For instance, $h = 0$ gets rewritten to $\exists t_0 \, .\, t_0 + g_{11} = 0 \land t_0 = h_2$, which itself gets rewritten to $\exists t_0 \, .\, t_0 + g_{11} = 0 \land (\exists t_1 \, .\, t_0 = t_1 + g_{21} \land t_1 = h_3)$. We use our key observation that the previously quantified $t_0$ does not occur in $t_1 = h_3$, and can be recycled for this purpose of rewriting it. The intermediate formula after completing the applications of Rule 1 recycles $t_0, t_1$ in an alternating manner while introducing quantified variables.

    We now have subformulas of the form $t_b = t_{1-b} + g_{r1}$ that we need to rewrite using Rule~2. We cannot use $t_0, t_1$, and hence must use $t_2, t_3$ to obtain $$\exists t_2, t_3 \, .\, t_b = t_{1-b} + t_2 - t_3 \land t_2 = (x + g_{r2})^2 \land t_3 = (-x + g_{r2})^2,$$
    to which Rule 2 may need to be reapplied. This time, however, we have access to $t_0, t_1$ while rewriting $t_2 = (x + g_{r2})^2$. In this manner, we can alternate between introducing $t_0, t_1$, and $t_2, t_3$ while applying Rule~2. This leaves us with subformulas of the form $t_a = (x + t_b - t_c)^2$ to rewrite using Rule 3. This is merely a syntactic rewrite, and we have indeed proven that we need only $4$ existentially quantified variables. 
\end{proof}

\subsection{Universal Diophantine Equations and Undecidability}
It is well known that Hilbert's tenth problem, i.e., deciding whether a given polynomial equation has integer solutions, is undecidable. Thanks to Xiao's proof of B\"uchi's conjecture~\cite{xiao2025buchi}, together with Lem.~\ref{lem:encode-dio} and the bounds on the degree of the polynomial and number of variables, we now establish undecidability results for bounded-variable Presburger arithmetic expanded with the perfect-square predicate. More specifically, Jones \cite{jonesuniversal} constructs \emph{universal} Diophantine equations, i.e., polynomials $h$ in several unknowns $x_1, \ldots, x_n$ and parameters $x, y, z, w \in \mathbb{N}$ such that $$\exists x_1, \ldots, x_n \in \mathbb{N} \, .\, h(x, y, z, w, x_1, \ldots, x_n) = 0$$ if and only if $x$ is contained in the recursively enumerable set indexed by $\langle y, z, w\rangle$.

\begin{theorem}
\label{thm:undec}
    The following theories are undecidable:
    \begin{enumerate}
        \item $\exists \mathit{FO}^{13}\langle \mathbb{Z}; 0, 1, +, -, <, \mathbf{Z}^2 \rangle$,
        
        \item $\exists \mathit{FO}^{14}\langle \mathbb{Z}; 0, 1, +, -, \mathbf{Z}^2 \rangle$,
        
        \item $\mathit{SMT}^{600}\langle \mathbb{Z}; 0, 1, +, -, <, \mathbf{Z}^2\rangle$,

        \item $\mathit{SMT}^{2200}\langle \mathbb{Z}; 0, 1, +, -, \mathbf{Z}^2\rangle$.
    \end{enumerate}
\end{theorem}

\begin{proof}
Item (1) follows from Matiyasevich's construction (see \cite[Sec.~3]{jonesuniversal}), which when given $x \in \mathbb{N}$ and a recursively enumerable set $W$, produces a Diophantine equation with $9$ positive-integer unknowns that has a solution if and only if $x \in W$. Item (2) follows from the analogue due to Sun \cite[Thm.~1.1(ii)]{sun-hilbert-Z}, where the constructed equation has $9$ integer unknowns and $1$ nonzero-integer unknown. 

Item (3) follows from Jones's concrete example of a universal Diophantine equation of degree $4$ with $58$ positive-integer unknowns, which can be implemented with at most $100$ arithmetic operations \cite[Thm.~5]{jonesuniversal}. Note that each multiplication would introduce at most $5$ fresh variables to be encoded in our SMT instance (as discussed in the proof of Lem.~\ref{lem:encode-dio}). We take $600 = 100 + 100\cdot 5$ as a conservative estimate for the total number of variables. 

Item (4) follows from converting each positive-integer unknown into a regular integer unknown by introducing fresh variables and using the Lagrange four-squares theorem, i.e., $\exists x > 0$ replaced by $\exists y_1, \ldots, y_4 \, . \bigwedge_{i=1}^4 \mathbf{Z}^2(y_i) \land \bigvee_{i=1}^4 y_i \ne 0$, and every occurrence of $x$ is replaced by $(y_1 + \cdots + y_4)$. A conservative upper bound on the number of variables introduced in this manner is $4\cdot 4\cdot 100$, and combining them with the original variables gives a sound estimate of $2200$.
\end{proof}
\bibliographystyle{plain}
\bibliography{main}

\end{document}